\documentclass[11pt]{article}
\usepackage{fullpage}
\usepackage{amsmath,amsfonts,amssymb,amsthm}
\usepackage{algorithm,algpseudocode,float}
\usepackage{xspace}
\usepackage[usenames]{color}

\newcommand{\supp}{\mathop{support}}
\newcommand{\suppfind}[1]{support-finding$_{{#1}}$}

\newtheorem{theorem}{Theorem}
\newtheorem{lemma}{Lemma}

\newtheorem{corollary}{Corollary}
\newtheorem{claim}{Claim}
\newtheorem{remark}{Remark}
\DeclareMathOperator*{\E}{\mathbb{E}}
\let\Pr\relax
\DeclareMathOperator*{\Pr}{\mathbb{P}}

\newcommand{\success}{\textsf{SUCC}\xspace}
\newcommand{\enc}{\textsf{ENC}\xspace}
\newcommand{\dec}{\textsf{DEC}\xspace}
\newcommand{\s}{\textsf{s}\xspace}
\newcommand{\R}{\mathbb{R}}
\newcommand{\F}{\mathbb{F}}

\newcommand{\sketch}{\mathsf{Alice}}
\newcommand{\query}{\mathsf{Bob}}
\newcommand{\eps}{\varepsilon}
\newcommand{\ur}{\mathbf{UR}\xspace}
\newcommand{\randcom}{\mathbf{R}}
\newcommand{\findup}[1]{\textsf{FindDuplicate}$({#1})$\xspace}

\title{Optimal lower bounds for universal relation, samplers, and finding duplicates}
\author{Jelani Nelson\thanks{Harvard University. \texttt{minilek@seas.harvard.edu}. Supported by NSF grant IIS-1447471 and
   CAREER award CCF-1350670, ONR Young Investigator award N00014-15-1-2388, and a Google Faculty Research Award.}
  \and Jakub Pachocki\thanks{OpenAI. \texttt{jakub@openai.com}. Work done while affiliated with Harvard University, under the support of ONR grant N00014-15-1-2388.}
  \and Zhengyu Wang\thanks{Harvard University. \texttt{zhengyuwang@g.harvard.edu}. Supported by NSF grant CCF-1350670.}}

\begin{document}

\setcounter{page}{0}

\maketitle

\thispagestyle{empty}

\begin{abstract}
In the communication problem $\ur$ ({\em universal relation}) \cite{KarchmerRW95}, Alice and Bob respectively receive $x$ and $y$ in $\{0,1\}^n$ with the promise that $x\neq y$. The last player to receive a message must output an index $i$ such that $x_i\neq y_i$. We prove that the randomized one-way communication complexity of this problem in the public coin model is exactly $\Theta(\min\{n, \log(1/\delta)\log^2(\frac{n}{\log(1/\delta)})\})$ bits for failure probability $\delta$.  We also show that for a more general problem $\ur_k$, in which the output must be $\min\{k, |\supp(x-y)|\}$ distinct indices $i$ such that $x_i\neq y_i$, the optimal randomized one-way communication complexity is $\Theta(\min\{n, t\log^2(n/t)\})$ bits where $t = \max\{k, \log(1/\delta)\}$. Our lower bounds hold even if promised $\supp(y)\subset \supp(x)$.

As a corollary, we obtain optimal lower bounds for sampling problems in strict turnstile streams, as well as for the problem of finding duplicates in a stream.  Specifically, consider a high-dimensional vector $x\in\R^n$ receiving streaming updates of the form ``$x_i\leftarrow x_i + \Delta$'', and in response to a query we must with probability $1-\delta$ recover {\em any} element $i\in\supp(x) = \{j : x_j \neq 0\}$. Our lower bound implies the first optimal $\Omega(\log(1/\delta)\log^2 n)$-bit space lower bound for any solution to this problem as long as $\delta > 2^{-n^{.99}}$, matching an upper bound of \cite{JowhariST11}. Our result thus implies optimal lower bounds for the so-called ``$\ell_p$-sampling'' problem for any $0\le p < 2$ in the strict turnstile model, as well as variations in which a query response must include $\min\{k, |\supp(x)|\}$ elements from $\supp(x)$.  Our lower bounds also do not need to use large weights, and hold even if it is promised that $x\in\{0,1\}^n$ at all points in the stream.

More easily explained in the language of the above support-finding turnstile streaming problem, our lower bound operates by showing that any algorithm $\mathcal{A}$ solving that problem in low memory can be used to encode subsets of $[n]$ of certain sizes into a number of bits below the information theoretic minimum. Our encoder makes adaptive queries to $\mathcal{A}$ throughout its execution, but done carefully so as to not violate correctness. This is accomplished by injecting random noise into the encoder's interactions with $\mathcal{A}$, which is loosely motivated by techniques in differential privacy. Our correctness analysis involves understanding the ability of $\mathcal{A}$ to correctly answer adaptive queries which have positive but bounded mutual information with $\mathcal{A}$'s internal randomness, and may be of independent interest in the newly emerging area of adaptive data analysis with a theoretical computer science lens.
\end{abstract}

\newpage

\section{Introduction}\label{sec:intro}
In turnstile $\ell_0$-sampling, a vector $z\in\R^n$ starts as the zero vector and receives coordinate-wise updates of the form ``$z_i \leftarrow z_i + \Delta$'' for $\Delta\in\{-M,-M+1,\ldots,M\}$. During a query, one must return a uniformly random element from $\supp(x) = \{i : z_i\neq 0\}$. The problem was first defined in \cite{FrahlingIS08}, where a data structure (or ``sketch'') for solving it was used to estimate the Euclidean minimum spanning tree, and to provide $\eps$-approximations of a point set $P$ in a geometric space (that is, one wants to maintain a subset $S\subset P$ such that for any set $R$ in a family of bounded VC-dimension, such as the set of all axis-parallel rectangles, $||R\cap S|/|S| - |R\cap P|/|P|| < \eps$). Sketches for $\ell_0$-sampling were also used to solve various dynamic graph streaming problems in \cite{AhnGM12a} and since then have been crucially used in almost all known dynamic graph streaming algorithms\footnote{The spectral sparsification algorithm of \cite{KapralovLMMS14} is a notable exception.}, such as for: connectivity, $k$-connectivity, bipartiteness, and minimum spanning tree \cite{AhnGM12a}, subgraph counting, minimum cut, and cut-sparsifier and spanner computation \cite{AhnGM12b}, spectral sparsifiers \cite{AhnGM13}, maximal matching \cite{ChitnisCHM15}, maximum matching \cite{AhnGM12a,BuryS15,Konrad15,AssadiKLY16,ChitnisCEHMMV16,AssadiKL17}, vertex cover \cite{ChitnisCHM15,ChitnisCEHMMV16}, hitting set, $b$-matching, disjoint paths, $k$-colorable subgraph, and several other maximum subgraph problems \cite{ChitnisCEHMMV16}, densest subgraph \cite{BhattacharyaHNT15,McGregorTVV15,EsfandiariHW16}, vertex and hyperedge connectivity \cite{GuhaMT15}, and graph degeneracy \cite{FarachColtonT16}. For an introduction to the power of $\ell_0$-sketches in designing dynamic graph stream algorithms, see the recent survey of McGregor \cite[Section 3]{McGregor14}. Such sketches have also been used outside streaming, such as in distributed algorithms \cite{HegemanPPSS15,Pandurangan0S16} and data structures for dynamic connectivity \cite{KapronKM13,Wang15,GibbKKT15}.

Given the rising importance of $\ell_0$-sampling in algorithm design, a clear task is to understand the exact complexity of this problem. The work \cite{JowhariST11} gave an $\Omega(\log^2 n)$-bit space lower bound for data structures solving the case $M=1$ which fail with constant probability, and otherwise whose query responses are $(1/3)$-close to uniform in statistical distance. They also gave an upper bound for $M \le \mathop{poly}(n)$ with failure probability $\delta$, which in fact gave $\min\{\|z\|_0, \Theta(\log(1/\delta))\}$ uniform samples from the support of $z$, using space $O(\log^2 n \log(1/\delta))$ bits (here $\|z\|_0$ denotes $|\supp(z)|$). Thus we say their data structure actually solves the harder problem of $\ell_0$-sampling$_{\Theta(\log(1/\delta))}$ with failure probability $\delta$, where in $\ell_0$-sampling$_k$ the goal is to recover $\min\{\|z\|_0, k\}$ uniformly random elements, without replacement, from $\supp(z)$.  The upper and lower bounds in \cite{JowhariST11} thus match up to a constant factor for $k = 1$ and $\delta$ a constant.

\paragraph{Universal relation.} The work of \cite{JowhariST11} obtains its lower bound for $\ell_0$-sampling (and some other problems) via reductions from {\em universal relation} ($\ur$). The problem $\ur$ was first defined in \cite{KarchmerRW95} and arose in connection with work of Karchmer and Wigderson on circuit depth lower bounds \cite{KarchmerW90}. For $f:\{0,1\}^n\rightarrow\{0,1\}$, $D(f)$ is the minimum depth of a fan-in $2$ circuit over the basis $\{\neg, \vee, \wedge\}$ computing $f$. Meanwhile, the (deterministic) communication complexity $C(f)$ is defined as the minimum number of bits that need to be communicated in a correct protocol for Alice and Bob to solve the following communication problem: Alice receives $x\in f^{-1}(0)$ and Bob receives $y\in f^{-1}(1)$ (and hence in particular $x\neq y$), and they must both agree on an index $i\in[n]$ such that $x_i\neq y_i$. It is shown in \cite{KarchmerW90} that $D(f) = C(f)$, where they then used this correspondence to show a tight $\Omega(\log^2 n)$ depth lower bound on monotone circuits solving undirected $s$-$t$ connectivity. The work of \cite{KarchmerRW95} then proposed a strategy to separate the complexity classes $\mathbf{NC}^1$ and $\mathbf{P}$: start with a function $f$ on $\log n$ bits requiring depth $\Omega(\log n)$, then ``compose'' it with itself $k = \log n / \log\log n$ times (see \cite{KarchmerW90} for a precise definition of composition). If one could prove a strong enough direct sum theorem for communication complexity after composition, such a $k$-fold composition would yield a function that is provably in $\mathbf{P}$ (and in fact, even in $\mathbf{NC}^2$), but not in $\mathbf{NC}^1$. Proving such a direct sum theorem is still wide open, and the statement that it is true is known as the ``KRW conjecture''; see for example the recent works \cite{GavinskyMWW14,DinurM16} toward resolving this conjecture. As a toy problem en route to resolving it, \cite{KarchmerRW95} suggested proving a direct sum theorem for $k$-fold composition of a particular function $\ur$ that they defined; that task was positively resolved in \cite{EIRS91} (see also \cite{HastadW90}.

The problem $\ur$ abstracts away the function $f$ and requires Alice and Bob to agree on the index $i$ only knowing that $x,y\in\{0,1\}^n$ are unequal. The deterministic communication complexity of $\ur$ is nearly completely understood, with upper and lower bounds that match up to an additive $3$ bits, even if one requires an upper bound on the number of rounds \cite{TardosZ97}. Henceforth we also consider a generalized problem $\ur_k$, where the output must be $\min\{k, \|x-y\|_0\}$ distinct indices on which $x, y$ differ. We also use $\ur^{\subset}, \ur_k^{\subset}$ to denote the variants when promised $\supp(y)\subset \supp(x)$, and also Bob knows $\|x\|_0$. Clearly $\ur, \ur_k$ can only be harder than $\ur^\subset, \ur_k^\subset$, respectively.

More than twenty years after its initial introduction in connection with circuit depth lower bounds, Jowhari et al.\ in \cite{JowhariST11} demonstrated the relevance of $\ur$ in the randomized one-way communication model for obtaining space lower bounds for certain streaming problems, such as various sampling problems and finding duplicates in streams. In particular, if $\randcom^{\rightarrow,pub}_\delta(f)$ denotes the randomized one-way communication complexity of $f$ in the public coin model with failure probability $\delta$, \cite{JowhariST11} showed that the space complexity of \findup{n} with failure probability $\delta$ is at least $\randcom^{\rightarrow,pub}_{\frac 78 + \frac{\delta}8}(\ur)$. In \findup{n}, one is given a length-$(n+1)$ stream of integers in $[n]$, and the algorithm must output some element $i\in[n]$ which appeared at least twice in the stream (note that at least one such element must exist, by the pigeonhole principle). The work \cite{JowhariST11} then showed a reduction demonstrating that any solution to $\ell_0$-sampling with failure probability $\delta$ in turnstile streams immediately implies a solution to \findup{n} with failure probability at most $(1+\delta)/2$ in the same space (and thus the space must be at least $\randcom^{\rightarrow,pub}_{\frac{15}{16} + \frac{\delta}{16}}(\ur)$). The same result is shown for $\ell_p$-sampling for any $p>0$, in which the output index should equal $i$ with probability $|x_i|^p/(\sum_j |x_j|^p)$, and a similar result is shown even if the distribution on $i$ only has to be close to this $\ell_p$-distribution in variational distance (namely, the distance should be bounded away from $1$). It is then shown in \cite{JowhariST11} that $\randcom^{\rightarrow,pub}_\delta(\ur) = \Omega(\log^2 n)$ for any $\delta$ bounded away from $1$. The approach used though unfortunately does not provide an improved lower bound for $\delta\downarrow 0$.

Seemingly unnoticed in \cite{JowhariST11}, we first point out here that the lower bound proof for $\ur$ in that work actually proves the same lower bound for the promise problem $\ur^\subset$. This observation has several advantages. First, it makes the reductions to the streaming problems trivial (they were already quite simple when reducing from $\ur$, but now they are even simpler). Second, a simple reduction from $\ur^\subset$ to sampling problems provides space lower bounds even in the strict turnstile model, and even for the simpler {\em \suppfind{}} streaming problem for which when queried is allowed to return {\em any} element of $\supp(z)$, without any requirement on the distribution of the index output. Both of these differences are important for the meaningfulness of the lower bound. This is because in dynamic graph streaming applications, typically $z$ is indexed by $\binom{n}{2}$ for some graph on $n$ vertices, and $z_e$ is the number of copies of edge $e$ in some underlying multigraph. Edges then are never deleted unless they had previously been inserted, thus only requiring sampler subroutines that are correct with the strict turnstile promise. Also, for every single application mentioned in the first paragraph of Section~\ref{sec:intro} (except for the two applications in \cite{FrahlingIS08}), the known algorithmic solutions which we cited as using $\ell_0$-sampling as a subroutine actually only need a subroutine for the easier \suppfind{} problem. Finally, third and most relevant to our current work's main focus, the straightforward reductions from $\ur^\subset$ to the streaming problems we are considering here do not suffer any increase in failure probability, allowing us to transfer lower bounds on $\randcom^{\rightarrow,pub}_{\delta}(\ur^\subset)$ for small $\delta$ to lower bounds on various streaming problems for small $\delta$. The work \cite{JowhariST11} could not provide lower bounds for the streaming problems considered there in terms of $\delta$ for small $\delta$.

We now show simple reductions from $\ur^\subset$ to \findup{n} and from $\ur_k^\subset$ to \suppfind{k}. In \suppfind{k} we must report $\min\{k,\|z\|_0\}$ elements in $\supp(z)$. In the claims below, $\delta$ is the failure probability for the considered streaming problem.

\begin{claim}
Any one-pass streaming algorithm for \findup{n} must use $\randcom^{\rightarrow,pub}_{\delta}(\ur^\subset)$ space.
\end{claim}
\begin{proof}
  We reduce from $\ur^\subset$. Suppose there were a space-$S$ algorithm $\mathcal{A}$ for \findup{n}. Alice creates a stream consisting of all elements of $\supp(x)$ and runs $\mathcal{A}$ on those elements, then sends the memory contents of $\mathcal{A}$ to Bob. Bob then continues running $\mathcal{A}$ on $n+1-\|x\|_0$ arbitrarily chosen elements of $[n]\backslash\supp(y)$. Then there must be a duplicate in the resulting concatenated stream, and all duplicates $i$ satisfy $x_i\neq y_i$.
\end{proof}

\begin{claim}
Any one-pass streaming algorithm for \suppfind{k} in the strict turnstile model must use $\randcom^{\rightarrow,pub}_{\delta}(\ur_k^\subset)$ bits of space, even if promised that $z\in\{0,1\}^n$ at all points in the stream.
\end{claim}
\begin{proof}
This is again via reduction from $\ur_k^\subset$. Let $\mathcal{A}$ be a space-$S$ algorithm for \suppfind{k} in the strict turnstile model. For each $i\in\supp(x)$, Alice sends the update $z_i \leftarrow z_i + 1$ to $\mathcal{A}$. Alice then sends the memory contents of $\mathcal{A}$ to Bob. Bob then for each $i\in\supp(y)$ sends the update $z_i\leftarrow z_i - 1$ to $\mathcal{A}$. Now note that $z$ is exactly the indicator vector of the set $\{i : x_i\neq y_i\}$.
\end{proof}

\begin{claim}
Any one-pass streaming algorithm for $\ell_p$-sampling for any $p\ge 0$ in the strict turnstile model must use $\randcom^{\rightarrow,pub}_{\delta}(\ur_k^\subset)$ bits of space, even if promised $z\in\{0,1\}^n$ at all points in the stream.
\end{claim}
\begin{proof}
This is via straightforward reduction from \suppfind{k}, since reporting $\min\{k,\|z\|_0\}$ elements of $\supp(z)$ satisfying some distributional requirements is only a harder problem than finding {\em any} $\min\{k,\|z\|_0\}$ elements of $\supp(z)$.
\end{proof}

The reductions above thus raise the question: what is the asymptotic behavior of $\randcom^{\rightarrow,pub}_\delta(\ur_k^\subset)$?

\paragraph{Our main contribution:} We prove for any $\delta$ bounded away from $1$ and $k\in[n]$, $\randcom^{\rightarrow,pub}_\delta(\ur_k^\subset) = \Theta(\min\{n, t\log^2(n/t)\})$ where $t = \max\{k,\log(1/\delta)\}$. Given known upper bounds in \cite{JowhariST11}, our lower bounds are optimal for \findup{n}, \suppfind{}, and $\ell_p$-sampling for any $0\le p<2$ for nearly the full range of $n, \delta$ (namely, for $\delta > 2^{-n^{.99}}$). Also given an upper bound of \cite{JowhariST11}, our lower bound is optimal for $\ell_0$-sampling$_k$ for nearly the full range of parameters $n, k, \delta$ (namely, for $t < n^{.99}$). Previously no lower bounds were known in terms of $\delta$ (or $k$). Our main theorem:

\begin{theorem}\label{thm:main}
For any $\delta$ bounded away from $1$ and $1\le k\le n$, $\randcom^{\rightarrow,pub}_\delta(\ur_k^\subset) = \Theta(\min\{n, t\log^2(n/t)\})$.
\end{theorem}

Our upper bound is also new, though follows by minor modifications of the upper bound in \cite{JowhariST11} and thus we describe it in the appendix. The previous upper bound was $O(\min\{n, t\log^2 n\})$. We also mention here that it is known that the upper bound for both $\ur_k$ and $\ell_0$-sampling$_k$ in two rounds (respectively, two passes) is only $O(t\log n)$ \cite{JowhariST11}. Thus, one cannot hope to extend our new lower bound to two or more passes, since it simply is not true.

\subsection{Related work}
The question of whether $\ell_0$-sampling is possible in low memory in turnstile streams was first asked in \cite{CormodeMR05,FrahlingIS08}. The work \cite{FrahlingIS08} was applied $\ell_0$-sampling as a subroutine in approximating the cost of the Euclidean minimum spanning tree of a subset $S$ of a discrete geometric space subject to insertions and deletions. The algorithm given there used space $O(\log^3 n)$ bits to achieve failure probability $1/\mathop{poly}(n)$ (though it is likely that the space could be improved to $O(\log^2 n\log\log n)$ with a worse failure probability, by replacing a subroutine used there with a more recent $\ell_0$-estimation algorithm of \cite{KaneNW10}). As mentioned, the currently best known upper bound solves $\ell_0$-sampling$_k$ using $O(t\log^2 n)$ bits \cite{JowhariST11}, which Theorem~\ref{thm:main} shows is tight.

For $\ell_p$-sampling as defined above, the first work to realize its importance came even earlier than for $\ell_0$-sampling: \cite{CoppersmithK04} showed that an $\ell_2$-sampler using small memory would lead to a nearly space-optimal streaming algorithm for multiplicatively estimating $\|x\|_3$ in the turnstile model, but did not know how to implement such a data structure. The first implementation was given in \cite{MonemizadehW10}, where they achieved space $\mathop{poly}(\eps^{-1}\log n)$ for failure probability $1/\mathop{poly}(n)$. For $1\le p\le 2$ the space was improved to $O(\eps^{-p}\log^3 n)$ bits for constant failure probability \cite{AndoniKO11}. In \cite{JowhariST11} this bound was improved to $O(\eps^{-\max\{1,p\}}\log(1/\delta)\log^2 n)$ bits for failure probability $\delta$ when $0<p<2$ and $p\neq 1$. For $p=1$ the space bound achieved by \cite{JowhariST11} was a $\log(1/\eps)$ factor worse: $O(\eps^{-1}\log(1/\eps)\log(1/\delta)\log^2 n)$ bits.

For finding a duplicate item in a stream, the question of whether a space-efficient randomized algorithm exists was asked in \cite{Muthukrishnan05,Tarui07}. The question was positively resolved in \cite{GopalanR09}, which gave an $O(\log^3 n)$-space algorithm with constant failure probability. An improved algorithm was given in \cite{JowhariST11}, using $O(\log(1/\delta) \log^2 n)$ bits of space for failure probability $\delta$.

\section{Overview of techniques}\label{sec:overview}
We describe our proof of Theorem~\ref{thm:main}. For the upper bound, \cite{JowhariST11} achieved $O(t\log^2n)$, but in the appendix we show that slight modifications to their approach yield $O(\min\{n,t\log^2(n/t)\})$ bits. Our main contribution is in proving an improved lower bound. Assume $t < cn$ for some sufficiently small constant $c$ (since otherwise we already obtain an $\Omega(n)$ lower bound). Our lower bound proof in this regime is split into two parts: we show $\randcom^{\rightarrow,pub}_\delta(\ur^\subset) = \Omega(\log \frac 1{\delta}\log^2 \frac n{\log\frac 1{\delta}})$ and $\randcom^{\rightarrow,pub}_{\frac 12}(\ur_k^\subset)=\Omega(k\log^2\frac nk)$ separately. We give an overview the former here, which is the more technically challenging half. Our proof of the latter can be found in Section~\ref{sec:k-samples-lb}.

We prove the lower bound via an encoding argument. Fix $m$. A randomized encoder is given a set $S\subset[n]$ with $|S| = m$ and must output an encoding $\enc(S)$, and a decoder sharing public randomness with the encoder must be able to recover $S$ given only $\enc(S)$. We consider such schemes in which the decoder must succeed with probability $1$, and the encoding length is a random variable. Any such encoding must use $\Omega(\log(^n_m)) = \Omega(m\log \frac nm)$ bits in expectation for some $S$.

There is a natural, but sub-optimal approach to using a public-coin one-way protocol $\mathcal{P}$ for $\ur^\subset$ to devise such an encoding/decoding scheme.  The encoder pretends to be Alice with input $x$ being the indicator set of $S$, then lets $\enc(S)$ be the message $M$ Alice would have sent to Bob. The decoder attempts to recover $S$ by iteratively pretending to be Bob $m$ times, initially pretending to have input $y=0\in\{0,1\}^n$, then iteratively adding elements found in $S$ to $y$'s support. Henceforth let $\mathbf{1}_T\in\{0,1\}^n$ denote the indicator vector of a set $T\subset[n]$.

\begin{algorithm}[H] 
  \caption{Simple Decoder.} \label{algo:wrong}
  \begin{algorithmic}[1]
    \Procedure{$\dec$}{$M$}
    \State $T\leftarrow \emptyset$
    \For {$r=1,\ldots,m$} 
      \State Let $i$ be Bob's output upon receiving message $M$ from Alice when Bob's input is $\mathbf{1}_T$
      \State $T \leftarrow T \cup\{i\}$
    \EndFor
    \State \Return $T$
    \EndProcedure
  \end{algorithmic}
\end{algorithm}

One might hope to say that if the original failure probability were $\delta < 1/m$, then by a union bound, with constant probability every iteration succeeds in finding a new element of $S$ (or one could even first apply some error-correction to $x$ so that the decoder could recover $S$ even if only a constant fraction of iterations succeeded). The problem with such thinking though is that this decoder chooses $y$'s adaptively! To be specific, $\mathcal{P}$ being a correct protocol means
\begin{equation}
\forall x,y\in\{0,1\}^n,\ \Pr_s(\mathcal{P}\text{ is correct on inputs }x,y) \ge 1-\delta , \label{eqn:correct}
\end{equation}
where $s$ is the public random string that both Alice and Bob have access to. The issue is that even in the second iteration (when $r=2$), Bob's ``input'' $\mathbf{1}_T$ {\em depends on $s$}, since $T$ depends on the outcome of the first iteration! Thus the guarantee of \eqref{eqn:correct} does not apply.

One way around the above issue is to realize that as long as every iteration succeeds, $T$ is always a subset of $S$. Thus it suffices for the following event $\mathcal{E}$ to occur: $\forall T\subset S,\ \mathcal{P}\text{ is correct on inputs }\mathbf{1}_S, \mathbf{1}_T$. Then $\Pr_s(\neg \mathcal{E}) \le 2^m\delta$ by a union bound, which is at most $1/2$ for $m = \lfloor \log_2(1/\delta)\rfloor - 1$. We have thus just shown that $\randcom^{\rightarrow,pub}_\delta(\ur^\subset) = \Omega(\min\{n, \log(^n_m)\}) = \Omega(\min\{n, \log\frac 1{\delta}\log \frac n{\log(1/\delta)}\})$.

Our improvement is as follows. Our new decoder again iteratively tries to recover elements of $S$ as before. We will give up though on having $m$ iterations and hoping for all (or even most) of them to succeed. Instead, we will only have $R = \Theta(\log \frac 1{\delta}\log \frac n{\log \frac 1{\delta}})$ iterations, and our aim is for the decoder to succeed in finding a new element in $S$ for at least a constant fraction of these $R$ iterations. Simplifying things for a moment, let us pretend for now that all $R$ iterations do succeed in finding a new element. $\enc(S)$ will then be Alice's message $M$, together with the set $B\subset S$ of size $m-R$ not recovered during the $R$ rounds, explicitly written using $\lceil\log{n \choose |B|}\rceil$ bits. If the decoder can then recover these $R$ remaining elements, this then implies the decoder has recovered $S$, and thus we must have $|M| = \Omega(\log{n\choose m} - \log{n \choose |B|}) = \Omega(R\log \frac nm)$. The decoder proceeds as follows. Just as before, initially the decoder starts with $T = \emptyset$ and lets $i$ be the output of Bob on $\mathbf{1}_T$ and adds it to $T$. Then in iteration $r$, before proceeding to the next iteration, the decoder randomly picks some elements from $B$ and adds them into $T$, so that the number of elements left to be uncovered is some fixed number $n_r$. These extra elements being added to $T$ should be viewed as ``random noise'' to mask information about the random string $s$ used by $\mathcal{P}$, an idea very loosely inspired by ideas in differential privacy. For intuition, as an example suppose the iteration $r=1$ succeeds in finding some $i\in S$. If the decoder were then to add $i$ to $T$, as well as $\approx m/2$ random elements from $B$ to $T$, then the resulting $T$ reveals only $\approx 1$ bit of information about $i$ (and hence about $s$). This is as opposed to the $\log n$ bits $T$ would have revealed if the masking were not performed. Thus the next query in round $r=2$, although correlated with $s$, has very weak correlation after masking and we thus might hope for it to succeed. This intuition is captured in the following lemma, which we prove in Section~\ref{sec:optimal-lb}:
\begin{lemma}\label{lem:information}
  Consider $f$: $\{0,1\}^b\times \{0,1\}^q\rightarrow \{0,1\}$ and $X\in\{0,1\}^b$ uniformly random. If $\forall y\in \{0,1\}^q,\ \Pr(f(X,y)=1)\le \delta$ where $0<\delta<1$, then for any random variable $Y$ supported on $\{0,1\}^q$,
  \begin{align}
    \Pr(f(X,Y)=1)\le \frac{I(X;Y)+H_2(\delta)}{\log \frac{1}{\delta}}, \label{eqn:adaptivity}
  \end{align}
  where $I(X;Y)$ is the mutual information between $X$ and $Y$, and $H_2$ is the binary entropy function.
\end{lemma}
Fix some $x\in\{0,1\}^n$. One should imagine here that $f(X,y)$ is $1$ iff $\mathcal{P}$ fails when Alice has input $x$ and Bob has input $y$ in a $\ur^\subset$ instance, and the public random string is $X=s$. Then the lemma states that if $y=Y$ is not arbitrary, but rather random (and correlated with $X$), then the failure probability of the protocol is still bounded as long as the mutual information between $X$ and $Y$ is bounded. It is also not hard to see that this lemma is sharp up to small additive terms. Consider the case $x,y\in[n]$, and $f(x,y) = 1$ iff $x = y$. Then if $X$ is uniform, for all $y$ we have $\Pr(f(X,y) = 1) = 1/n$. Now consider the case where $Y$ is random and equal to $X$ with probability $t/\log n$ and is uniform in $[n]$ with probability $1 - t/\log n$. Then in expectation $Y$ reveals $t$ bits of $X$, so that $I(X;Y) = t$. It is also not hard to see that $\Pr(f(X,Y) = 1) \approx t/\log n + 1/n$.

In light of the strategy stated so far and Lemma~\ref{lem:information}, the path forward is clear: at each iteration $r$, we should add enough random masking elements to $T$ to keep the mutual information between $T$ and all previously added elements below, say, $\frac 12 \log \frac 1{\delta}$. Then we expect a constant fraction of iterations to succeed. The encoder knows which iterations do not succeed since it shares public randomness with the decoder (and can thus simulate it), so it can simply tell the decoder which rounds are the failed ones, then explicitly include in $M$ correct new elements of $S$ for the decoder to use in the place of Bob's wrong output in those rounds. A calculation shows that if one adds a $(1-1/K)\approx 2^{-1/K}$ fraction of the remaining items in $S$ to $T$ after drawing one more support element from Bob, the mutual information between the next query to Bob and the randomness used by $\mathcal{P}$ will be $O(K)$ (see Lemma~\ref{lemma:mutual-entropy-bound}). Thus we do this for $K$ a sufficiently small constant times $\log \frac 1{\delta}$. We will then have $n_r \approx (1 - 1/K)^r m$. Note that we cannot continue in this way once $n_r < K$ (since the number of ``random noise'' elements we inject should at least be one). Thus we are forced to stop after $R = \Theta(K\log(m/K)) = \Theta(\log\frac 1{\delta} \log\frac n{\log \frac 1{\delta}})$ iterations. We then set $m = \sqrt{n\log(1/\delta)}$, so that $\randcom^{\rightarrow,pub}_\delta(\ur^\subset) = \Omega(|R|\log \frac nm) = \Omega(\min\{n, \log\frac 1{\delta}\log^2 \frac n{\log \frac 1{\delta}}\})$ as desired.

The argument for lower bounding $\randcom^{\rightarrow,pub}_\delta(\ur_k^\subset)$ is a bit simpler, and in particular does not need rely on Lemma~\ref{lem:information}. Both the idea and rigorous argument can be found in Section~\ref{sec:k-samples-lb}, but again the idea is to use a protocol for this problem to encode appropriately sized subsets of $[n]$.
 
As mentioned above, our lower bounds use protocols for $\ur^\subset$ and $\ur^\subset_k$ to establish protocols for encoding subsets of some fixed size $m$ of $[n]$. These encoders always consist of some message $M$ Alice would have sent in a $\ur^\subset$ or $\ur^\subset_k$ protocol, together with a random subset $B\subset S$ (using $\lceil \log_2|B|\rceil + \lceil\log{n\choose |B|}\rceil$ bits, to represent both $|B|$ and the set $B$ itself). Here $|B|$ is a random variable. These encoders are thus {\em Las Vegas}: the length of the encoding is a random variable, but the encoder/decoder always succeed in compressing and recovering the subset. The final lower bounds then come from the following simple lemma, which follows from the source coding theorem. 

\begin{lemma} \label{lemma:lb-meta}
  Let $\s$ denote the number of bits used by the $\ur^\subset$ or $\ur^\subset_k$ protocol, and let $\s'$ denote the expected number of bits to represent $B$. Then $(1+\s+\s') \ge \log (^n_m)$. In particular, $s \ge \log(^n_m) - s' - 1$.
\end{lemma}

Section~\ref{sec:optimal-lb} provides the full details of the proof that $\randcom^{\rightarrow,pub}_\delta(\ur^\subset) = \Omega(\min\{n, \log^2(\frac n{\log(1/\delta)}) \log \frac{1}{\delta}\})$. We extend our results in Section~\ref{sec:k-samples-lb} to $\ur_k^\subset$ for $k\ge 1$, proving a lower bound of $\Omega(k\log^2(n/k))$ communication even for constant failure probability.

\section{Communication Lower Bound for $\ur^\subset$} \label{sec:optimal-lb}

Consider a protocol $\mathcal{P}$ for $\ur^\subset$ with failure probability $\delta$, operating in the one-way public coin model. When Alice's input is $x$ and Bob's input is $y$, Alice sends $\sketch(x)$ to Bob, and Bob outputs $\query(\sketch(x), y)$, which with probability at least $1-\delta$ is in $\supp(x-y)$. As mentioned in Section~\ref{sec:overview}, we use $\mathcal{P}$ as a subroutine in a scheme for encoding/decoding elements of $\binom{[n]}m$ for $m = \lfloor \sqrt{n\log(1/\delta)}\rfloor$. In this section, we assume $\log \frac 1{\delta} \le n/64$, since for larger $n$ we have an $\Omega(n)$ lower bound.

\subsection{Encoding/decoding scheme}
We now describe our encoding/decoding scheme $(\enc, \dec)$ for elements in ${[n] \choose m}$, which uses $\mathcal{P}$ in a black-box way. The parameters shared by $\enc$ and $\dec$ are given in Algorithm~\ref{algo:para}.

\begin{algorithm}[H] 
  \caption{Variables Shared by encoder $\enc$ and decoder $\dec$.} \label{algo:para}
  \begin{algorithmic}[1] 
    \State $m\leftarrow \lfloor \sqrt{n \log\frac{1}{\delta}} \rfloor$ 
    \State $K\leftarrow \lfloor \frac{1}{16}\log \frac{1}{\delta} \rfloor$
    \State $R\leftarrow \lfloor K\log(m/4K) \rfloor$
    \For {$r = 0, \ldots, R$}
      \State $n_r\leftarrow \lfloor m \cdot 2^{-\frac{r}{K}} \rfloor$ \Comment{$|S_r|=n_r$, and we have $n_r-n_{r+1}\ge 2$}
    \EndFor
    \State Let $\pi$ be a random permutation on $[n]$ \Comment{Used to generate $S_r$ and $C_r$}
  \end{algorithmic}
\end{algorithm}

As discussed in Section~\ref{sec:overview}, on input $S\in {[n] \choose m}$, $\enc$ computes $M \leftarrow \sketch(\mathbf{1}_S)$ as part of its output. Moreover, $\enc$ also outputs a subset $B\subseteq S$ computed as follows. Initially $B=S$ and $S_0=S$. $\enc$ proceeds in $R$ rounds.  In round $r\in[R]$, $\enc$ computes $s_r\leftarrow \query(M, \mathbf{1}_{S\backslash S_{r-1}})$.  Let $b$ denote a binary string of length $R$, where $b_r$ records whether $\query$ succeeds in round $r$.  $\enc$ also outputs $b$.  If $s_r\in S_{r-1}$, i.e. $\query(M, \mathbf{1}_{S\backslash S_{r-1}})$ succeeds, $\enc$ sets $b_r=1$ and removes $s_r$ from $B$ (since the decoder can recover $s_r$ from the $\ur^\subset$-protocol, $\enc$ does not need to include it in $B$); otherwise $\enc$ sets $b_r=0$.  At the end of round $r$, $\enc$ picks a uniformly random set $S_r$ in $\binom{S_{r-1}\backslash \{s_r\}}{n_r}$.  In particular, $\enc$ uses its shared randomness with $\dec$ to generate $S_r$ in such a way that $\enc, \dec$ agree on the sets $S_r$ ($\dec$ will actually iteratively construct $C_r = S\backslash S_r$). We present $\enc$ in Algorithm~\ref{algo:enc}.

The decoding process is symmetric.  Let $C_0=\emptyset$ and $A=\emptyset$.  $\dec$ proceeds in $R$ rounds.  On round $r\in[R]$, $\dec$ obtains $s_r\in S\backslash C_{r-1}$ by invoking $\query(M, \mathbf{1}_{C_{r-1}})$.  By construction of $C_{r-1}$ (to be described later), it is guaranteed that $S_{r-1}=S\backslash C_{r-1}$.  Therefore, $\dec$ recovers exactly the same $s_r$ as $\enc$.  $\dec$ initially assigns $C_r\leftarrow C_{r-1}$.  If $b_r=1$, $\dec$ adds $s_r$ to both $A$ and $C_r$.  At the end of round $r$, $\dec$ inserts many random items from $B$ into $C_r$ so that $C_r=S\backslash S_r$.  $\dec$ can achieve this because of the shared random permutation $\pi$ when constructing $S_r$.  In the end, $\dec$ outputs $B\cup A$.  We present $\dec$ in Algorithm~\ref{algo:dec}.

\begin{algorithm}[H] 
  \caption{Encoder $\enc$.} \label{algo:enc}
  \begin{algorithmic}[1]
    \Procedure{$\enc$}{$S$}
    \State $M \leftarrow \sketch(\mathbf{1}_S)$
    \State $A\leftarrow \emptyset$
    \State $S_0 \leftarrow S$
    \For {$r=1,\ldots,R$}
      \State $s_r\leftarrow \query(M, \mathbf{1}_{S\backslash S_{r-1}})$
      \State $S_r\leftarrow S_{r-1}$
      \If {$s_r\in S_{r-1}$} \Comment{i.e. if $s_r$ is a valid sample}
        \State $b_r\leftarrow 1$ \Comment{$b$ is a binary string of length $R$, indicating if $\query$ succeeds on round $r$}
        \State $A\leftarrow A \cup \{s_r\}$
        \State $S_r\leftarrow S_r \backslash \{s_r\}$
      \Else 
        \State $b_r\leftarrow 0$
      \EndIf
      \State Remove $|S_r|-n_r$ elements from $S_r$ with smallest $\pi_a$'s among $a\in S_r$ \Comment{So that $|S_r|=n_r$}
    \EndFor
    \State \Return ($M$, $S\backslash A$, $b$) 
    \EndProcedure
  \end{algorithmic}
\end{algorithm}

\begin{algorithm}[H] 
  \caption{Decoder $\dec$.} \label{algo:dec}
  \begin{algorithmic}[1]
    \Procedure{$\dec$}{$M$, $B$, $b$}
    \State $A\leftarrow \emptyset$
    \State $C_0 \leftarrow \emptyset$
    \For {$r=1,\ldots,R$}
      \State $C_r\leftarrow C_{r-1}$
      \If{$b_r=1$}
        \State $s_r\leftarrow \query(M, \mathbf{1}_{C_{r-1}})$ \Comment{Invariant: $C_r=S \backslash S_r$ ($S_r$ is defined in $\enc$)}
        \State $A\leftarrow A \cup \{s_r\}$
        \State $C_r\leftarrow C_r \cup \{s_r\}$
      \EndIf
       \State Insert $m-n_r-|C_r|$ items into $C_r$ with smallest $\pi_a$'s among $a\in B\backslash C_r$
    \EndFor
    \State \Return $B\cup A$ 
    \EndProcedure
  \end{algorithmic}
\end{algorithm}

\subsection{Analysis}

We have two random objects in our encoding/decoding scheme: (1) the random source used by $\mathcal{P}$, denoted by $X$, and (2) the random permutation $\pi$. These are independent.

First, we can prove that $\dec(\enc(S))=S$.  That is, for any fixing of the randomness in $X$ and $\pi$, $\dec$ will always decode $S$ successfully.  It is because $\enc$ and $\dec$ share $X$ and $\pi$, so that $\dec$ essentially simulates $\enc$.  We formally prove this by induction in Lemma~\ref{lemma:zero-fail-prob}.

Now our goal is to prove that by using the $\ur^\subset$-protocol, the number of bits that $\enc$ saves in expectation over the naive $\lceil\log(^n_m)\rceil$-bit encoding is $\Omega(\log \frac{1}{\delta}\log^2 \frac{n}{\log (1/\delta)} )$ bits.  Intuitively, it is equivalent to prove the number of elements that $\enc$ saves is $\Omega(\log \frac{1}{\delta}\log \frac{n}{\log (1/\delta)} )$.
We formalize this in Lemma~\ref{lemma:bits-saving}. 
Note that $\enc$ also needs to output $b$ (i.e., whether the $\query$ succeeds on $R$ rounds), which takes $R$ bits. 
By our setting of parameters, we can afford the loss of $R$ bits.  Thus it is sufficient to prove $\E|B|=|S|-\Omega(\log \frac{1}{\delta}\log \frac{n}{\log (1/\delta)})$. 

We have $|S|-|B|=\sum_{r=1}^{R}b_r$. 
In Lemma~\ref{lem:information}, we prove the probability that $\query$ fails on round $r$ is upper bounded by $\frac{I(X;S_{r-1})+1}{\log \frac{1}{\delta}}$, where $I(X;S_{r-1})$ is the mutual information between $X$ and $S_{r-1}$. 
Furthermore, we will show in Lemma~\ref{lemma:mutual-entropy-bound} that $I(X;S_{r-1})$ is upper bounded by $O(K)$.
By our setting of parameters, we have $\E b_r=\Omega(1)$ and thus $\E(|S|-|B|)=\Omega(R)=\Omega(\log \frac{1}{\delta}\log \frac{n}{\log (1/\delta)})$.
 
\begin{lemma}\label{lemma:zero-fail-prob}
  $\dec(\enc(S))=S$.
\end{lemma}
\begin{proof}
  We claim that for $r=0,\ldots, R$, $\{S_r, C_r\}$ is a partition of $S$ ($S_r$ is defined in Algorithm~\ref{algo:enc}, and $C_r$ in Algorithm~\ref{algo:dec}). We prove the claim by induction on $r$. Our base case is $r=0$, for which the claim holds since $S_0 = S$, $C_0 = \emptyset$.
  
  Assume the claim holds for $r-1$ ($1\le r \le R$), and we consider round $r$.  On round $r$, by induction $S\backslash S_{r-1}=C_{r-1}$, the index $s_r$ obtained by both \enc and \dec are the same.  Initially $S_r=S_{r-1}$ and $C_r=C_{r-1}$, and so $\{S_r,C_r\}$ is a partition of $S$.  If $s_r$ is a valid sample (i.e. $s_r\in S_{r-1}$), then $b_r=1$, and \enc removes $s_r$ from $S_r$ and in the meanwhile \dec inserts $s_r$ into $C_r$, so that $\{S_r, C_r\}$ remains a partition of $S$. Next, \enc repeats removing the $a$ from $S_r$ with the smallest $\pi_a$ value until $|S_r|=n_r$. Symmetrically, \dec repeats inserting the $a$ into $C_r$ with the smallest $\pi_a$ value among $a\in B\backslash C_r$, until $|C_r|=|S|-n_r$. In the end we have $|S_r|+|C_r|=|S|$, so \enc and \dec execute repetition the same number of times.  Moreover, we can prove that during the same iteration of this repeated insertion, the element removed from $S_r$ is exactly the same element inserted to $C_r$.  This is because in the beginning of a repetition $\{S_r, C_r\}$ is a partition of $S$.  We have $B\backslash C_r\subseteq S\backslash C_r=S_r$. Let $a^*$ denote $a\in S_r$ that minimizes $\pi_a$.  Then $a^*\in B\backslash C_r\subseteq S_r$ (since $a^*$ will be removed from $S_r$, it has no chance to be included in $S$ in \enc, so that $B$ contains $a^*$), and $\pi_{a^*}$ is also the smallest among $\{\pi_a : a\in B\backslash C_r\}$.  Thus both $\enc$ and $\dec$ will take $a^{*}$ (for \enc, to remove from $S_r$, and for \dec, to insert into $C_r$).  Therefore, $\{S_r, C_r\}$ remains a partition of $S$.
  
  Given the fact that $\{S_r, C_r\}$ is a partition of $S$, the $s_r$ are the same in \enc and \dec.  Furthermore, $A=\{s_r : b_r=1,r=1,\ldots, R\}$ are the same in \enc and \dec.  We know $A\subseteq S$.  Since \enc outputs $S\backslash A$, and \dec outputs $(S\backslash A)\cup A$, we have $\dec(\enc(S))=S$.
\end{proof}

\begin{lemma} \label{lemma:bits-saving}
Let $W\in \mathbb{N}$ be a random variable with $W\le m$ and $\E W\le m-d$. Then $\E(\log {n \choose m}-\log {n \choose W})\ge d \log (\frac{n}{m}-1)$.
\end{lemma}

\begin{proof}
  \begin{align*}
  \log {n \choose m}-\log {n \choose W}
  &= \log \frac{n!/(m!(n-m)!)}{n!/(W!(n-W)!)} \\
  &= \sum_{i=1}^{m-W}\log \frac{n-W-i+1}{m-i+1} \\
  &\ge (m-W)\cdot \log \frac{n-W}{m} \\
  &\ge (m-W)\cdot \log \frac{n-m}{m}
  \end{align*}
  
  Taking expectation on both sides, we have $\E(\log {n \choose m}-\log {n \choose W})\ge d \log (\frac{n}{m}-1)$. 
\end{proof}

\noindent \textbf{Lemma~\ref{lem:information} (restated).}
  Consider $f$: $\{0,1\}^b\times \{0,1\}^q\rightarrow \{0,1\}$ and $X\in\{0,1\}^b$ uniformly random. If $\forall y\in \{0,1\}^q,\ \Pr(f(X,y)=1)\le \delta$ where $0<\delta<1$, then for any r.v.\ $Y$ supported on $\{0,1\}^q$,
$$
  \Pr(f(X,Y)=1)\le \frac{I(X;Y)+H_2(\delta)}{\log \frac{1}{\delta}} ,
$$
  where $I(X;Y)$ is the mutual information between $X$ and $Y$, and $H_2$ is the binary entropy function.
\begin{proof}
  It is equivalent to prove 
$$I(X;Y)\ge \E(f(X,Y))\cdot \log\frac{1}{\delta}-H_2(\delta).$$
By definition of mutual entropy $I(X;Y)=H(X)-H(X|Y)$, where $H(X)=b$ and we must show
$$H(X|Y)\le H_2(\delta)+(1-\E(f(X,Y)))\cdot b+\E(f(X,Y))\cdot (b-\log\frac{1}{\delta})=b+H_2(\delta)-\E(f(X,Y))\cdot \log\frac{1}{\delta} .$$
  The upper bound for $H(X|Y)$ is obtained by considering the following one-way communication problem: Alice knows both $X$ and $Y$ while Bob only knows $Y$, and Alice must send a single message to Bob so that Bob can recover $X$. The expected message length in an optimal protocol is exactly $H(X|Y)$.  Thus, any protocol gives an upper bound for $H(X|Y)$, and we simply take the following protocol: Alice prepends a $1$ bit to her message iff $f(X,Y) = 1$ (taking $H_2(\delta)$ bits in expectation). Then if $f(X,Y)=0$, Alice sends $X$ directly (taking $b$ bits). Otherwise, when $f(X,Y)=1$, Alice sends the index of $X$ in $\{x|f(x,Y)=1\}$ (taking $\log (\delta 2^b)=b-\log\frac{1}{\delta}$ bits).  
\end{proof}

\begin{corollary}\label{corollary:sampler-failure}
  Let $X$ denote the random source used by the $\ur^\subset$-protocol with failure probability at most $\delta$. If $S$ is a fixed set and $T\subset S$, $\Pr(\query(\sketch(\mathbf{1}_S), \mathbf{1}_T)\not\in S\backslash T)\le \frac{I(X;T)+H_2(\delta)}{\log\frac{1}{\delta}}$.
\end{corollary}

\begin{lemma}\label{lemma:mutual-entropy-bound}
  $I(X;S_r)\le 6K$, for $r=1,\ldots, R$.
\end{lemma}

\begin{proof}
  Note that $I(X;S_r)=H(S_r)-H(S_r|X)$. Since $|S_r|=n_r$ and $S_r\subseteq S$, $H(S_r)\le \log {m \choose n_r}$. Here is the main idea to lower bound $H(S_r|X)$: By definition of conditional entropy, $H(S_r|X)=\sum_x{p_x\cdot H(S_r|X=x)}$. We fix an arbitrary $x$. If we can prove that for any $T\subseteq S$ where $|T|=n_r$, $\Pr(S_r=T|X=x)\le p$, then by definition of entropy we have $H(S_r|X=x)\ge\log\frac{1}{p}$. 
  
  First we can prove for any fixed $T$,
  
  \begin{align}
    \Pr(S_r=T|X=x)\le \prod_{i=1}^{r}{\frac{{n_{i-1}-n_r-1 \choose n_{i-1}-n_i-1}}{{n_{i-1}-1 \choose n_{i-1}-n_i-1}}}.
  \end{align}
  
  We have $\Pr(S_r=T|X=x)=\Pi_{i=1}^{r}{\Pr(T\subseteq S_i|T\subseteq S_{i-1})}$. 
  On round $i$ ($1\le i \le r$), $\enc$ removes $n_{i-1}-n_i$ elements (at least $n_{i-1}-n_i-1$ of which are chosen all at random) from $S_{i-1}$ to obtain $S_i$. 
  Conditioned on the event that $T\subseteq S_{i-1}$, the probability that $T\subseteq S_i$ is at most ${{n_{i-1}-n_r-1 \choose n_{i-1}-n_i-1}}/{{n_{i-1}-1 \choose n_{i-1}-n_i-1}}$, where the equation achieves when $s_i\in S_{i-1}\backslash T$, and $\enc$ takes a uniformly random subset of $S_{i-1}\backslash \{s_i\}$ of size $n_{i-1}-n_i-1$, so that the subset does not intersect with $T$.
  
  Next we can prove 
  
  \begin{align}
    \prod_{i=1}^{r}{\frac{{n_{i-1}-n_r-1 \choose n_{i-1}-n_i-1}}{{n_{i-1}-1 \choose n_{i-1}-n_i-1}}} \le \frac{2^{6K}}{{m \choose n_r}}. \label{eqn:prod-bound}
  \end{align}
    
  For notational simplicity, let $n^{\underline{k}}$ denote $n\cdot (n-1)\ldots (n-k+1)$. We have 
  \begin{align}
    \prod_{i=1}^{r}{\frac{{n_{i-1}-n_r-1 \choose n_{i-1}-n_i-1}}{{n_{i-1}-1 \choose n_{i-1}-n_i-1}}}
    =\prod_{i=1}^{r}\frac{(n_{i-1}-n_r-1)!n_i!}{(n_{i-1}-1)!(n_i-n_r)!}
    =\prod_{i=1}^{r}\frac{n_i^{\underline{n_r}}}{(n_{i-1}-1)^{\underline{n_r}}}
    =\prod_{i=1}^{r} \left( \frac{n_i^{\underline{n_r}}}{n_{i-1}^{\underline{n_r}}}\cdot \frac{n_{i-1}}{n_{i-1}-n_r} \right).
  \end{align}
  
  By telescoping,
  \begin{align}
    \prod_{i=1}^{r} \frac{n_i^{\underline{n_r}}}{n_{i-1}^{\underline{n_r}}}
    =\frac{n_r^{\underline{n_r}}}{n_0^{\underline{n_r}}}
    =\frac{n_r!(n_0-n_r)!}{n_0!}=\frac{1}{{n_0 \choose n_r}}
    =\frac{1}{{m \choose n_r}}.
    \label{eqn:telescoping}
  \end{align}
  
  Moreover, 
  \begin{align}
    \prod_{i=1}^{r} \frac{n_{i-1}}{n_{i-1}-n_r}
    \le\prod_{i=1}^{r} \frac{1}{1-\frac{m\cdot 2^{-r/K}}{m\cdot 2^{-(i-1)/K}-1}}
    \le\prod_{i=1}^{r} \frac{1}{1-\frac{m\cdot 2^{-r/K}+1}{m\cdot 2^{-(i-1)/K}}}
    =\prod_{j=1}^{r} \frac{1}{1-2^{-j/K}-\frac{2^{\frac{r-j}{K}}}m}.
    \label{eqn:bound-with-floor}
  \end{align}
  
  By our setting of parameters 
  $$\frac{2^{\frac rK}}m \le \frac{2^{\frac RK}}m \le \frac{1}{4K} .$$
  
  Therefore, for $j\in \{1,\ldots, r\}$,
  $$\frac{1}{1-2^{-\frac jK}-\frac{2^{\frac{r-j}{K}}}m}\le \frac{1}{1-(1+\frac{1}{4K})2^{-\frac jK}}.$$ 
  
  By Taylor series $2^{1/K} = \sum_{n=0}^{\infty}{\frac{(\ln 2 )^n}{n!K^n}} >1+\frac{\ln 2}{K}>1+\frac{1}{4K}$, and thus $\frac{1}{1-(1+\frac{1}{4K})2^{-j/K}}\le \frac{1}{1-2^{(1-j)/K}}$, for $j=2,\ldots, r$. For $j=1$, we have $\frac{1}{1-(1+\frac{1}{4K})2^{-\frac 1K}} \le 2^K$.
  
  By Lemma~\ref{lemma:Pochhammer}, we have $\prod_{j=1}^{\infty} \frac{1}{1-2^{-j/K}}\le 2^{5K}$. Therefore, the right hand side of \eqref{eqn:bound-with-floor} is upper bounded by $2^{6K}$. Together with \eqref{eqn:telescoping}, we prove \eqref{eqn:prod-bound} holds.  
  
  Finally, let $p={2^{6K}}/{{m\choose n_r}}$, we have $\Pr(S_r=T|X=x)\le p$ and thus $H(S_r|X=x)\ge \log\frac{1}{p}=\log{{m\choose n_r}}-6K$. Therefore, $H(S_r|X)\ge \log{{m\choose n_r}}-6K$ and so $I(X;S_r)=H(S_r)-H(S_r|X)\le 6K$.  
\end{proof}

\begin{lemma}\label{lemma:Pochhammer}
  Let $K\in \mathbb{N}$ and $K\ge 1$. We have $\prod_{j=1}^{\infty} \frac{1}{1-2^{-j/K}}\le 2^{5K}$.
\end{lemma}

\begin{proof}
  First, we bound the product of first $2K$ terms. Note that $\frac{1}{1-2^{-x}}\le \frac{8}{3x}$ for $0<x\le 2$. Therefore, 
  \begin{align}
    \prod_{j=1}^{2K}\frac{1}{1-2^{-j/K}}
    \le (8/3)^{2K}\cdot \frac{K^{2K}}{(2K)!}
    \le (8/3)^{2K}\cdot \frac{K^{2K}}{(2K/e)^{2K}}
    = (4e/3)^{2K}
    < 2^{4K}. 
  \end{align}
  
  Then, we bound the product of the rest terms
  \begin{align}
    \prod_{j=2K+1}^{\infty}\frac{1}{1-2^{-j/K}} 
    \le \prod_{j=2K+1}^{\infty}\frac{1}{1-2^{-\lfloor j/K \rfloor}} 
    \le \prod_{i=2}^{\infty}\left( \frac{1}{1-2^{-i}}\right)^K 
    \le \left( \frac{1}{1-\sum_{i=2}^{\infty}2^{-i}}\right)^K
    = 2^K.
  \end{align}
  
  Multiplying two parts proves the lemma.
\end{proof}

\begin{theorem}
  $\randcom^{\rightarrow,pub}_\delta(\ur^\subset) = \Omega(\log \frac{1}{\delta}\log^2 \frac{n}{\log (1/\delta)} )$, given that $64 \le \log \frac{1}{\delta} \le \frac{n}{64}$.
\end{theorem}

\begin{proof}
  By Lemma~\ref{lemma:zero-fail-prob}, the success probability of protocol $(\enc,\dec)$ is $1$. 
  By Lemma~\ref{lemma:lb-meta}, we have $\s\ge \log (^n_m) - \s' -1$, where $\s'=\log n + R+ \E(\log (^n_{|B|}))$. 
  The size of $B$ is $|B|=|S|-\sum_{r=1}^{R}{b_r}$.
  By Corollary~\ref{corollary:sampler-failure}, conditioned on $S$, $\Pr(b_r=0)\le \frac{I(X;S_{r-1})+1}{\log\frac{1}{\delta}}$. 
  By Lemma~\ref{lemma:mutual-entropy-bound}, $I(X;S_{r-1})\le 6K$ (Note that when $r=1$, $I(X;S_0)=0\le 6K$). 
  Therefore, $\E(b_r)\ge 1-\frac{6K+1}{\log\frac{1}{\delta}}$.
  By the setting of parameters (see Algorithm~\ref{algo:para}) we have $\E(b_r)\ge \frac{39}{64}$. Therefore, $\E(|B|)\le |S|-\frac{39}{64}R$. 
  By Lemma~\ref{lemma:bits-saving}, $\log (^n_m)-\E(\log (^n_{|B|}))\ge \frac{39}{64}R\cdot \log (\frac{n}{m}-1) \ge \frac{1}{2}R\log (\frac{n}{\log(1/\delta)})$. 
  Furthermore, $\frac{1}{6}R\log \frac{n}{\log (1/\delta)} \ge R$.
  Thus we obtain $\s \ge \frac{R}{3}\log \frac{n}{\log(1/\delta)} -(\log n + 1)  =\Omega(\log \frac{1}{\delta}\log^2 \frac{n}{\log (1/\delta)} )$.
\end{proof}

\section{Communication Lower Bound for $\ur_k^\subset$}\label{sec:k-samples-lb}

In this section, we prove the lower bound $\randcom^{\rightarrow,pub}_{1/2}(\ur^\subset_k) = \Omega(\min\{n, k\log^2 \frac{n}{k}\})$. In fact, our lower bound holds for any failure probability $\delta$ bounded away from $1$. Let $\mathcal{P}$ denote a $\ur_k^\subset$-protocol where Alice sends $\sketch_k(x)$ to Bob, and Bob outputs $\query_k(\sketch_k(x), y)$.  We consider the following encoding/decoding scheme $(\enc_k, \dec_k)$ for $S\in {[n] \choose m}$.  $\enc_k$ computes $M\leftarrow \sketch_k(\mathbf{1}_S)$ as part of its message. In addition, $\enc_k$ includes $B\subseteq S$ constructed as follows, spending $\lceil\log{n\choose |B|}\rceil$ bits.  Initially $B= S$, and $\enc_k$ proceeds in $R=\Theta(\log (n/k))$ rounds.  Let $S_0=S\supseteq S_1\supseteq \ldots \supseteq S_R$ where $S_r$ is generated by sub-sampling each element in $S_{r-1}$ with probability $\frac{1}{2}$.  In round $r$ ($r=1,\ldots, R$), $\enc_k$ tries to obtain $k$ elements from $S_{r-1}$ by invoking $\query_k(M, \mathbf{1}_{S\backslash S_{r-1}})$, denoted by $A_k$, and removes $A_k\cap (S_{r-1}\backslash S_{r})$ (whose expected size is $\frac{k}{2}$) from $B$.  Note that $\dec_k$ is able to recover the elements in $A_k\cap (S_{r-1}\backslash S_{r})$.  For each round the failure probability of $\query_k$ is at most $\delta$.  Thus we have $\E(|S|-|B|)\ge \frac{k}{2}\cdot (1-\delta) \cdot R=\Omega(k\log\frac{n}{k})$.  Furthermore, each element contains $\Theta(\log \frac{n}{k})$ bits of information, thus yielding a lower bound of $\Omega(k\log^2\frac{n}{k})$ bits.

In this section we assume $k \le n/2^{10}$, since for larger $n$ we have an $\Omega(n)$ lower bound.

\subsection{Encoding/decoding scheme}
\begin{algorithm}[H] 
  \caption{Variables Shared by Encoder $\enc_k$ and Decoder $\dec_k$.} \label{algo:para4}
  \begin{algorithmic}[1] 
    \State $m\leftarrow \lfloor \sqrt{nk} \rfloor$
    \State $R\leftarrow \lfloor \frac{1}{2}\log (n/k) - 2 \rfloor$ \Comment{Note that $R\ge 3$ because $k\le \frac{n}{2^{10}}$}
    \State $T_0\leftarrow [n]$
    \For {$r = 1, \ldots, R$}
      \State $T_r\leftarrow \emptyset$
      \State For each $a\in T_{r-1}$, $T_r\leftarrow T_r\cup \{a\}$ with probability $\frac{1}{2}$ \Comment{We have $S_r=S\cap T_r$}
    \EndFor
  \end{algorithmic}
\end{algorithm}

\begin{algorithm}[H] 
  \caption{Encoder $\enc_k$.} \label{algo:enc4}
  \begin{algorithmic}[1]
    \Procedure{$\enc_k$}{$S$}
    \State $M \leftarrow \sketch_k(\mathbf{1}_S)$
    \State $A\leftarrow \emptyset$
    \For {$r=1,\ldots,R$}
    \State $A_r\leftarrow \query_k(M, \mathbf{1}_{S\backslash (S\cap T_{r-1})})$
    \If {$A_r\subseteq S\cap T_{r-1}$} \Comment{i.e. if $A_r$ is valid}
      \State $b_r\leftarrow 1$ \Comment{$b$ is a binary string of length $R$, indicating if $\query_k$ succeeds in round $r$}
      \State $A\leftarrow A \cup (A_r\cap (T_{r-1}\backslash T_r))$
    \Else 
      \State $b_r\leftarrow 0$
    \EndIf
    \EndFor
      \State \Return ($M$, $S\backslash A$, $b$) 
    \EndProcedure
  \end{algorithmic}
\end{algorithm}

\begin{algorithm}[H] 
  \caption{Decoder $\dec_k$.} \label{algo:dec4}
  \begin{algorithmic}[1]
    \Procedure{$\dec_k$}{$M$, $B$, $b$}
    \State $A\leftarrow \emptyset$
    \State $C_0 \leftarrow \emptyset$
    \For {$r=1,\ldots,R$}
      \State $C_r\leftarrow C_{r-1}$
      \If {$b_r=1$}
        \State $A_r\leftarrow \query_k(M, \mathbf{1}_{C_{r-1}})$ \Comment{Invariant: $C_r=S\backslash (S\cap T_r)$}
        \State $A\leftarrow A \cup (A_r\cap (T_{r-1}\backslash T_r))$
        \State $C_r\leftarrow C_r \cup (A_r\cap (T_{r-1}\backslash T_r))$
      \EndIf
      \State $C_r\leftarrow C_r \cup (B\cap (T_{r-1}\backslash T_r))$
    \EndFor
    \State \Return $B\cup A$ 
    \EndProcedure
  \end{algorithmic}
\end{algorithm}

\subsection{Analysis}

\begin{theorem}
  $\randcom^{\rightarrow,pub}_\delta(\ur_k^\subset) = \Omega(k\log^2 \frac{n}{k} )$, given that $1 \le k \le \frac{n}{2^{10}}$ and $\delta \le \frac{1}{2}$.
\end{theorem}

\begin{proof}
Let $S_r=S\cap T_r$.  Let $\success$ denote the event that $|S\cap T_R|=|S_R|\ge k$.  Note that $\E|S_R|=\frac{1}{2^R}m=4k$. By the Chernoff bound, $\Pr(\success)\ge \frac{1}{2}$.  In the following, we argue conditioned on $\success$. Namely, in each round $r$, there are at least $k$ items in $S_r$.
  
Similar to Lemma~\ref{lemma:zero-fail-prob}, we can prove the protocol $(\enc_k,\dec_k)$ always succeeds.  By Lemma~\ref{lemma:lb-meta}, we have $\s\ge \log (^n_m) - \s' -2$, where $\s'=\log n + R+ \E \log (^n_{|B|})$.  The size of $B$ is $|B|=|S|-\sum_{r=1}^{R}{(b_r \cdot |A_r \cap (S_{r-1}\backslash S_r)|)}$.  The randomness used by $\mathcal{P}$ is independent from $S\backslash S_{r-1}$ for every $r\in[R]$.  Therefore, $\E b_r\ge 1-\delta\ge \frac{1}{2}$, and $b_r$ is independent from $|A_r \cap (S_{r-1}\backslash S_r)|$.  We have $\E|A_r \cap (S_{r-1}\backslash S_r)|=\frac{k}{2}$, and thus $\E(|S|-|B|)\ge \frac{kR}{4}$.  By Lemma~\ref{lemma:bits-saving}, $\log (^n_m)-\E\log (^n_{|B|})\ge \frac{kR}{4}\cdot \log (\frac{n}{m}-1) \ge \frac{kR}{9}\log (\frac{n}{k})$.  Moreover, $\frac{kR}{10}\log \frac{n}{k}\ge R$.  Thus we have $\s = \Omega(kR\log\frac{n}{k}) = \Omega(k\log^2 \frac{n}{k} )$.
\end{proof}

\section*{Acknowledgments}
Initially the authors were focused on proving optimal lower bounds for samplers, but we thank Vasileios Nakos for pointing out that our $\ur^\subset$ lower bound immediately implies a tight lower bound for finding a duplicate in data streams as well. Also, initially our proof of Lemma~\ref{lem:information} incurred an additive $1$ in the numerator of the right hand side of \eqref{eqn:adaptivity}. This is clearly suboptimal for small $I(X; Y)$ (for example, consider $I(X; Y) = 0$, in which case the right hand side should be $\delta$ and not $1/\log(1/\delta)$)). We thank T.S.\ Jayram for pointing out that a slight modification of our proof could actually replace the additive $1$ with the binary entropy function (and also for showing us a different proof of this lemma, which resembles the standard proof of Fano's inequality).

\bibliographystyle{alpha}

\newcommand{\etalchar}[1]{$^{#1}$}

\appendix

\section{Appendix}

\subsection{A tight upper bound for $\randcom^{\rightarrow,pub}_\delta(\ur_k)$}

In \cite[Proposition 1]{JowhariST11} it is shown that $\randcom^{\rightarrow,pub}_\delta(\ur_k) = O(\min\{n,t\log^2 n\})$ for $t = \max\{k,\log(1/\delta)\}$. Here we show that a minor modification of their protocol in fact shows the correct complexity $\randcom^{\rightarrow,pub}_\delta(\ur_k) = O(\min\{n,t\log^2(n/t)\})$, which given our new lower bound, is optimal up to a constant factor for the full range of $n,k,\delta$ as long as $\delta$ is bounded away from $1$.

Recall Alice and Bob receive $x, y\in\{0,1\}^n$, respectively, and share a public random string. Alice must send a single message $M$ to Bob, from which Bob must recover $\min\{k, \|x-y\|_0\}$ indices $i\in[n]$ for which $x_i\neq y_i$. Bob is allowed to fail with probability $\delta$. The fact that $\randcom^{\rightarrow,pub}_\delta(\ur_k) \le n$ is obvious: Alice can simply send the message $M = x$, and Bob can then succeed with failure probability $0$. We thus now show $\randcom^{\rightarrow,pub}_{e^{-ck}}(\ur_k) \le k\log^2(n/k)$ for some constant $c>0$, which completes the proof of the upper bound. We assume $k\le n/2$ (otherwise, Alice sends $x$ explicitly).

As mentioned, the protocol we describe is nearly identical to one in \cite{JowhariST11} (see also \cite{CormodeF14}). We will describe the new protocol here, then point out the two minor modifications that improve the $O(k\log^2 n)$ bound to $O(k\log^2(n/k))$ in Remark~\ref{rem:recov}. We first need the following lemma.

\begin{lemma}\label{lem:sparse-recov}
Let $\F_q$ be a finite field and $n>1$ an integer. Then for any $1\le k\le \frac n2$, there exists $\Pi_k\in \F_q^{m\times n}$ for $m = O(k\log_q(qn/k))$ s.t.\ for any $w\neq w'\in\F_q^n$ with $\|w\|_0, \|w'\|_0 \le k$, $\Pi_k w \neq \Pi_k w'$.
\end{lemma}
\begin{proof}
The proof is via the probabilistic method. $\Pi_k w = \Pi_k w'$ iff $\Pi_k (w - w') = 0$. Note $v = w-w'$ has $\|v\|_0 \le 2k$. Thus it suffices to show that such a $\Pi_k$ exists with no $(2k)$-sparse vector in its kernel. The number of vectors $v\in\F_q^n$ with $\|v_0\| \le 2k$ is at most $\binom{n}{2k}\cdot q^{2k}$. For any fixed $v$, $\Pr(\Pi_k v = 0) = q^{-m}$. Thus 
$$\Pr(\exists v, \|v\|_0 \le 2k: \Pi_k v = 0) \le \binom{n}{2k}\cdot q^{2k} \cdot q^{-m}$$ 
by a union bound. The above is strictly less than $1$ for $m > 2k + \log_q\binom{n}{2k}$, yielding the claim.
\end{proof}

\begin{corollary}\label{cor:ksparse}
Let $\F_q$ be a finite field and $n>1$ an integer. Then for any $1\le k\le \frac n2$, there exists $\Pi_k\in \F_q^{m\times n}$ for $m = O(k\log_q(qn/k))$ together with an algorithm $\mathcal{R}$ such that for any $w\in\F_q^n$ with $\|w\|_0 \le k$, $\mathcal{R}(\Pi_k w) = w$.
\end{corollary}
\begin{proof}
Given Lemma~\ref{lem:sparse-recov}, a simple such $\mathcal{R}$ is as follows. Given some $y = \Pi_k w^*$ with $\|w^*\|_0 \le k$, $\mathcal{R}$ loops over all $w$ in $\F_q^n$ with $\|w\|_0 \le k$ and outputs the first one it finds for which $\Pi_k w = y$.
\end{proof}

The protocol for $\ur_k$ is now as follows. Alice and Bob use public randomness to pick commonly known random functions $h_0,\ldots,h_L:[n]\rightarrow\{0,1\}$ for $L = \lfloor\log_2(n/k)\rfloor$, such that for any $i\in[n]$ and for any $j$, $\Pr(h_j(i) = 1) = 2^{-j}$. They also agree on a matrix $\Pi_{16k}$ and $\mathcal{R}$ as described in Corollary~\ref{cor:ksparse} for a sufficiently large constant $C>0$ to be determined later, with $q = 3$. Thus $\Pi_{16k}$ has $m = O(k\log(n/k))$ rows. Alice then computes $v_j = \Pi_{16k} x|_{h_j^{-1}(1)}$ for $j=0,\ldots,L$ where $v_j\in\F_q^m$, and her message to Bob is $M = (v_0,\ldots,v_L)$. For $S\subseteq [n]$ and $x$ an $n$-dimensional vector, $x|_S$ denotes the $n$-dimensional vector with $(x|_S)_i = x_i$ for $i\in S$, and $(x|_S)_i = 0$ for $i\notin S$. Note Alice's message $M$ is $O(k\log^2(n/k))$ bits, as desired. Bob then executes the following algorithm and outputs the returned values.

\begin{algorithm}[H] 
  \caption{Bob's algorithm in the $\ur_k$ protocol.} \label{algo:bob-protocol}
  \begin{algorithmic}[1]
    \Procedure{Bob}{$v_0,\ldots,v_L$}
    \For {$j=L,L-1,\ldots,0$}
      \State $v_j \leftarrow v_j - \Pi_{16k} y|_{h_j^{-1}(1)}$
      \State $w_j\leftarrow \mathcal{R}(v_j)$
      \If {$\|w_j\|_0 \ge k$ or $j=0$}
      \State \Return an arbitrary $\min\{k, \|w_j\|_0\}$ elements from $\supp(w_j)$
      \EndIf
    \EndFor
    \EndProcedure
  \end{algorithmic}
\end{algorithm}

The correctness analysis is then as follows, which is nearly the same as the $\ell_0$-sampler of \cite{JowhariST11}. If Alice's input is $x$ and Bob's is $y$, let $a = x-y \in \{-1,0,1\}^n$, so that $a$ can be viewed as an element of $\F_3^n$. Also let $a_j = a|_{h_j^{-1}(1)}$. Then $\E \|v_j\|_0 = \|a\|_0\cdot 2^{-j}$, and since $0\le \|a\|_0 \le n$, there either (1) exists a unique $0\le j^*\le L$ such that $2k\le \E\|a_j\|_0\cdot 2^{-j^*}< 4k$, or (2) $\|a\|_0 < 2k$ (in which case we define $j^* = 0$). Let $\mathcal{E}$ be the event that $\|a_j\|_0 \le 16k$ simultaneously for all $j\le j^*$. Let $\mathcal{F}$ be the event that {\it either} we are in case (2), or we are in case (1) and $\|a_{j^*}\|_0 \ge k$ holds. Note that conditioned on $\mathcal{E}, \mathcal{F}$ both occurring, Bob succeeds by Corollary~\ref{cor:ksparse}.

We now just need to show $\Pr(\neg\mathcal{E} \wedge \neg\mathcal{F}) < e^{-\Omega(k)}$. We use the union bound. First, consider $\mathcal{F}$. If $j^* = 0$, then $\Pr(\neg\mathcal{F}) = 0$. If $j^*\neq 0$, then $\Pr(\neg\mathcal{F}) \le \Pr(\|a_{j^*}\|_0 < \frac 12 \cdot\E\|a_{j^*}\|_0)$, which is $e^{-\Omega(k)}$ by the Chernoff bound since $\E\|a_{j^*}\|_0 = \Theta(k)$. Next we bound $\Pr(\neg \mathcal{E})$. For $j\ge j^*$, we know $\E\|a_j\|_0 \le 4k/2^{j-j^*}$. Thus, letting $\mu$ denote $\E\|a_j\|_0$, 
\begin{equation}
\Pr(\|a_j\|_0 > 16k) < \left(\frac{e^{\frac{16k}{\mu} - 1}}{(\frac{16k}{\mu})^{\frac{16k}{\mu}}}\right)^\mu < \left(\frac{16k}{\mu}\right)^{-\Omega(k)} < (e^{-Ck})^{j-j^*}\label{eqn:geometric}
\end{equation}
for some constant $C>0$ by the Chernoff bound and the fact that $16k/\mu \ge 4 > e$. Recall that the Chernoff bound states that for $X$ a sum of independent Bernoullis,
$$
\forall \delta > 0,\ \Pr(X > (1+\delta) \E X) < \left(\frac{e^\delta}{(1+\delta)^{1+\delta}}\right)^{\E X} .
$$
Then by a union bound over $j\ge j^*$ and applying \eqref{eqn:geometric},
$$
\Pr(\neg \mathcal{E}) = \Pr(\exists j\ge j^*: \|a_j\|_0 > 16k) < \sum_{j=j^*}^\infty (e^{-Ck})^{j-j^*} = O(e^{-Ck}) .
$$

\begin{remark}\label{rem:recov}
\textup{
As already mentioned, the protocol given above and the one described in \cite{JowhariST11} using $O(k\log^2 n)$ bits differ in minor points. First: the protocol there used $\lfloor\log_2 n\rfloor$ different hash functions $h_j$, but as seen above, only $\lfloor \log_2(n/k)\rfloor$ are needed. This already improves one $\log n$ factor to $\log(n/k)$. The other improvement comes from replacing the $k$-sparse recovery structure with $2k$ rows used in \cite{JowhariST11} with our Corollary~\ref{cor:ksparse}. Note the matrix $\Pi_k$ in our corollary has even {\it more} rows, but the key point is that the bit complexity is improved. Whereas using a $k$-sparse recovery scheme as described in \cite{JowhariST11} would use $2k$ linear measurements of a $k$-sparse vector $w\in\{-1,0,1\}^n$ with $\log n$ bits per measurement (for a total of $O(k\log n)$ bits), we use $O(k\log(n/k))$ measurements with only $O(1)$ bits per measurement. The key insight is that we can work over $\F_3^n$ instead of $\R^n$ when the entries of $w$ are in $\{-1,0,1\}$, which leads to our slight improvement.
}
\end{remark}

\end{document}